\documentclass[graybox]{svmult}

\usepackage{amsmath}
\usepackage{bbm}
\usepackage[T1]{fontenc}

\usepackage{type1cm}        
\usepackage{makeidx}         
\usepackage{graphicx}        
\usepackage{multicol}        
\usepackage[bottom]{footmisc}

\usepackage{newtxtext}       %
\usepackage[varvw]{newtxmath}       

\makeindex             

\begin{document}

\title*{Solving the initial value problem for cellular automata by pattern decomposition}
\author{Henryk Fukś\inst{1}\orcidID{0000-0001-5855-3706}}
\institute{ Department of Mathematics and Statistics\\
         Brock University,  St. Catharines, Canada\
\email{hfuks@brocku.ca}\\
\url{https://lie.ac.brocku.ca/$\sim$hfuks/} }
%
%
\maketitle

\abstract{For many cellular automata, it is possible to express the state of a given cell
after $n$ iterations as an explicit function of the initial configuration. We say that for such rules the solution of the  initial value problem can be obtained.
In some cases, one can construct the solution formula for the  initial value problem by analyzing the spatiotemporal pattern generated by the rule and decomposing it into
simpler segments which one can then  describe algebraically. We show an
example of a rule  when such approach is successful, namely elementary rule 156. 
Solution of the initial value problem for this rule is constructed and then
used to compute the density of ones after $n$ iterations, starting from
a random initial condition. We also show how to obtain probabilities of occurrence of
longer blocks of symbols.}

\section{Introduction}
For cellular automata (CA), similarly as for partial differential equations \cite{hadamard1923},
one can consider the Cauchy problem or the initial value problem (IVP) \cite{paper59}. 
Consider an elementary rule with the local function $f:\{0,1\}^3 \to \{0,1\}$.
Let $x=\{0,1\}^\mathbbm{Z}$ be the initial configuration, and define
the corresponding global function $F: \{0,1\}^\mathbbm{Z} \to \{0,1\}^\mathbbm{Z}$
by
$
[F(x)]_i=f(x_{i-1}, x_i, x_{i+1}).
$
The value of $[F^n(x)]_i$ is then usually called ``the state of the cell $i$ after $n$ iterations''. In what follows, we will use the convention that the uppercase letters $F,G,H,\ldots$ correspond to the global
functions of CA with local functions, respectively, $f,g,h\ldots$.

The initial value problem for CA is the problem of finding  an explicit expression
for the  state of cell $i$ after $n$ iterations as a function of the initial
configuration $x$. In other words, it is the problem of expressing
$[F^n(x)]_i$ in terms of components of the initial configuration $x$.
For elementary rules $[F^n(x)]_i$ can only depend on $x_{i-n},x_{i-n+1},\ldots, x_{i+n}$,
thus we are seeking to express $[F^n(x)]_i$ explicitly in terms of
$x_{i-n},x_{i-n+1},\ldots, x_{i+n}$.

Knowing the explicit formula for $[F^n(x)]_i$ is useful if we want to know what a given rule ``computes''. For example, elementary rule
 with Wolfram number 128 computes the product
of cell values \cite{hfbook}, and the IVP for this rule has the solution
$$
[F^n_{128}(x)]_j =\prod_{i=-n}^{n} x_{i+j}.
$$
Note that we indicated the Wolfram number of the rule by putting it in the index of
$F$, this is the convention we will be using in the remainder of the paper.
Although there exists a rule computing the product, not all  simple 
computations are possible with cellular automata. For example, we know  that there is no binary CA which would
compute the majority of cell values \cite{land95}, thus there is no CA with global function $F$ 
which would yield
$$
[F^n(x)]_j=\mathrm{majority}(x_{j-n}, x_{j-n+1}, \ldots, x_{j+n}).
$$
Therefore,  while products  can be easily computed,  majority cannot. Which functions
can then be computed by CA and which not? This is a very difficult question
for which the answer is not known, but by solving the IVP for other 
elementary rules we will at least be able to find out what kind of functions these
rules are capable of computing. In \cite{hfbook}, the IVP problem is solved for over 60 elementary CA. Some solutions presented there are very simple, 
like for rule 34,
\begin{equation*}
[F_{34}^n(x)]_j = -x_{{j+n}}x_{{j+n-1}}+x_{{j+n}}.
\end{equation*}
In other cases the solution formulae are moderately complicated, like for rule 172, originally analyzed in \cite{paper59},
\begin{multline*}
[F_{172}^n(x)]_j =  \overline{x}_{{j-2}} \overline{x}_{{j-1}}  x_{{j}}+
 \left(  \overline{x}_{{j+n-2}}  x_{{j+n-1}}+x_{{j+n-2}}x_{{j+n}}
 \right) \prod _{i=j-2}^{j+n-3}(1- \overline{x}_{{i}}  \overline{x}
_{{i+1}} ),
\end{multline*}
where $\overline{x}=1-x$.
There are also rules with solutions much more 
complex than the above, with multiple terms involving summations and products.

In this paper we will demonstrate  how to obtain the solution formula for one of such rules, using a technique of pattern decomposition.
We will use elementary rule 156 as an example. This example has been selected because even though the
rule is relatively simple in terms of its dynamics, the resulting
expression for $F$, as we will shortly see, is quite long, probably approaching
the limit of complexity which can be handled by this method.
The expression for $[F_{156}^n(x)]_j$ was given in~\cite{hfbook}, but many details of its derivation, and, most importantly, proof of its correctness was not included there.
\section{Motivation}
Cellular automata are often viewed as fully discrete analogs of partial differential
equations. While a wealth of methods for solving both ordinary and partial differential
equations exist, it is usually not possible to use any of such methods for CA due to the very
different structure of the underlying space. Nevertheless, some general ideas used in the theory of differential equations can be useful for cellular automata.
One of such ideas is the method for solving non-homogeneous linear equations such as
\begin{equation}\label{ode}
\dot{ \mathbf{x}}(t)=\mathbf{A x}(t) + \mathbf{b}(t),
\end{equation}
where $\mathbf{x}(t)$ is the unknown vector function, $\mathbf{A}$ is an is $n \times  n$ matrix and $\mathbf{b}(t)$ is a continuous vector valued function. One can treat the right hand side of the above
equation as a sum of ``unperturbed'' term to which a ``perturbation'' is added,
$$\dot{ \mathbf{x}}(t)=\underbrace{\mathbf{A x}(t)}_{\mathrm{unperturbed}} + \underbrace{ \mathbf{b}(t).}_{\mathrm{perturbation}}$$
 The solution of the ``unperturbed'' equation,
 $$\dot{ \mathbf{x}}(t)=\mathbf{A x}(t),$$
 is easy to find,
 $$\mathrm{x}(t)=e^{\mathbf{A}t} \mathbf{x}_0,$$
and then  the solution of the complete eq. (\ref{ode})
can be expressed  \cite{perko1996differential} as
 $$
 \mathrm{x}(t)=e^{\mathbf{A}t} \mathbf{x}_0
 +
 e^{\mathbf{A}t}
  \int_0^t e^{\mathbf{A}\tau} \mathbf{b}(\tau) d \tau.
 $$
 This solution can be interpreted as a sum of the solution of the ``unperturbed'' equation,
 $e^{\mathbf{A}t} \mathbf{x}_0$,
 and the effect of the perturbation, $e^{\mathbf{A}t}
  \int_0^t e^{\mathbf{A}\tau} \mathbf{b}(\tau) d \tau$.
 
Let us then suppose that the local function $f$ of an elementary CA can be decomposed into
 two parts,
 $$f(x_1,x_2,x_3)= g(x_1,x_2,x_3)+b(x_1,x_2,x_3),$$
 where $g(x_1,x_2,x_3)$ is the local function of some other CA with known
 formula for $[G^n(x)]_j$. The function $b(x_1,x_2,x_3)$ plays the role of a ``perturbation'' here. In analogy to the above differential equation example,
 we expect that the solution formula for $[F^n(x)]_j$ will be given
 by 
 $$[F^n(x)]_j=[G^n(x)]_j + P(x,n,j),$$
 where $P(x,n,j)$ is the effect of the perturbation $b(x_1,x_2,x_3)$.
 Obviously this is where the analogy stops, as we do not know
 any general method for finding $P(x,n,j)$ for a given $b(x_1,x_2,x_3)$.
 Nevertheless, in many cases the form of $P(x,n,j)$ can be obtained
 heuristically by analyzing the spatiotemporal diagram produced by the CA rule $f$
 and then verifying rigorously that such heuristic ``guess'' is correct.
 
 The example we are going to use in this paper are rules 156 and 140,
 for which $b(x_1,x_2,x_3)$ has a very simple form, being nonzero only in one case, namely
 $$f_{156}(x_1,x_2,x_3)=f_{140}(x_1,x_2,x_3)+b(x_1,x_2,x_3),$$
 $$b(x_1,x_2,x_3)=\begin{cases}
  1  &  \text{ if $(x_1,x_2,x_3)=(1,0,0)$,} \\
  0 & \text{ otherwise.}
\end{cases}$$
Furthermore, the formula for $[F_{140}^n(x)]_j$ can be constructed by induction,
as demonstrated in the next section.
\section{Rule 140}
The local function of rule 140 can be expressed as
$$f_{140}(x_1,x_2,x_3) = \begin{cases}
  1  &  \text{ if $(x_1,x_2,x_3)=(0,1,0), (0,1,1)$ or  $(1,1,1)$, } \\
  0 & \text{ otherwise.}
\end{cases}$$
If $x_1,x_2,x_3$ are Boolean variables, that is, taking values in $\{0,1\}$, then
it is easy to see that $f_{140}(x_1,x_2,x_3)$ returns 1 only when
one of the products 
$\overline{x}_1 x_2 \overline{x}_3$,
$\overline{x}_1 x_2 {x}_3$ or
$x_1 x_2 x_3$ is equal to 1, where where $\overline{x}=1-x$. We can, therefore, write
\begin{equation}\label{r140def}
f_{140}(x_1,x_2,x_3)=\overline{x}_1 x_2 \overline{x}_3+\overline{x}_1 x_2 {x}_3+
x_1 x_2 x_3 =\overline{x}_1 x_2+x_1 x_2 x_3.
\end{equation}
The last equality reflects the fact that $\overline{x}_3+x_3=1$.

Representation of the local function such as in eq. (\ref{r140def}) is called 
\emph{density polynomial representation}. A more formal definition of density polynomials can be found in \cite{hfbook}.
For now it suffices to say that once we have the density polynomial representation
of the local function $f$ of a cellular automaton,
we can obtain values of the cell $j$ in consecutive iterations of the cellular automaton
starting from any $x\in \{0,1\}^\mathbb{Z}$, as follows:
\begin{align*}
[F(x)]_j& = f(x_{j-1},x_j,x_{j+1})\\
[F^2(x)]_j& =
f\big(
f(x_{j-2},x_{j-1},x_{j  }),
f(x_{j-1},x_{j  },x_{j+1}),
f(x_{j  },x_{j+1},x_{j+2})
\big)\\
[F^3(x)]_j& =
f\Big(
 f\big(f(x_{j-3},x_{j-2},x_{j-1}), f(x_{j-2},x_{j-1},x_{j }), f(x_{j-1 },x_{j  },x_{j+1})\big)\\
&f\big(f(x_{j-2},x_{j-1},x_{j  }), f(x_{j-1},x_{j  },x_{j+1}), f(x_{j  },x_{j+1},x_{j+2})\big)\\
&f\big(f(x_{j-1},x_{j  },x_{j+1}), f(x_{j  },x_{j+1},x_{j+2}), f(x_{j+1},x_{j+2},x_{j+3})\big)
\Big)\\
&\ldots
\end{align*}
As remarked earlier, $[F^n(x)]_j$ is a function of 
$x_{j-n}, x_{j-n+1}, \ldots, x_{j+n}$. It will be convenient, therefore, to
define the function $f^n$ of $2n+1$ variables representing the functional dependence of
$[F^n(x)]_j$ on these variables, as follows:
$$
  f^n(x_1,x_2,\ldots x_{2n+1})=[F^n(x)]_{n+1}.
$$
We will call $f^n$ the \emph{$n$-th iterate of $f$}.

It is often possible to discover a pattern in formulae for $f^n$, and we will illustrate this
using rule 140 as an example.
For rule 140, its second iterate $f^2$ is given by
\begin{multline}
f^2_{140}(x_1,x_2,x_3,x_4,x_5)=
f_{140}(f_{140}(x_1,x_2,x_3),f_{140}(x_2,x_3,x_4),f_{140}(x_3,x_4,x_5))\\
=\overline{f}_{140}(x_1,x_2,x_3)f_{140}(x_2,x_3,x_4)
+f_{140}(x_1,x_2,x_3) f_{140}(x_2,x_3,x_4) f_{140}(x_3,x_4,x_5)\\
=(1-\overline{x}_1 x_2-x_1x_2x_3)(\overline{x}_2 x_3+x_2x_3x_4)\\
+
(\overline{x}_1 x_2+x_1 x_2 x_3) 
(\overline{x}_2 x_3+x_2 x_3 x_4)
(\overline{x}_3 x_4+x_3 x_4 x_5).
\end{multline}
Using the fact that for $x\in \{0,1\}$ we have $x^2=x$ and $x\overline{x}=0$, this simplifies to
$$f^2_{140}(x_1,x_2,x_3,x_4,x_5)=\overline{x}_2x_3 + x_2 x_3 x_4 x_5.$$
Similar (but more tedious) calculations for the third iterate yield
$$f_{140}^3(x_1,x_2,x_3,x_4,x_5,c_6,x_7)
=\overline{x}_{3}x_{4}+x_{3}x_{4}x_5 x_6 x_{7}.$$
Looking on the patterns which is developing, one can easily guess the general formula, 
$$f_{140}^n(x_1,x_2,\ldots, x_{2n+1})
=\overline{x}_{n}x_{n+1}+x_{n}x_{n+1}\ldots x_{2n+1}
=\overline{x}_{n}x_{n+1}+\prod_{i=n+1}^{2n+1} x_i.
$$
The ``guessed'' solution of the initial value problem, therefore,
can be formally expressed as follows.
\begin{theorem}
For elementary cellular automaton rule 140, 
for any $x\in \{0,1\}^{\mathbb{Z}}$ and $n>1$, the state of the $j$-th cell after $n$ iterations of the rule starting from $x$ is given by
\begin{equation}\label{r140sol}
[F_{140}^n(x)]_j=\overline{x}_{j-1}x_{j}+\prod_{i=n-1}^{2n} x_{i-n+j}.
\end{equation} 
\end{theorem}
Of course, at this point it not really a theorem but only a conjecture, thus we need to prove it,
meaning that we need to verify that the solution given in eq. (\ref{r140sol})
is indeed correct. There are two ways of doing this.
We can prove that for $x\in \{0,1\}^\mathbb{Z}$  
\begin{equation} \label{rule140abinitio}
[F_{140}^{n+1}(x)]_j=[F_{140}^{n}(y)]_j,
\end{equation}
where $y=F(x)$. This will be called verification \emph{ab initio}.
Alternatively, we can show that
\begin{equation} \label{rule140abfinito}
[F_{140}^{n+1}(x)]_j=f_{140}\left(F_{140}^{n}(x)]_{j-1},F_{140}^{n}(x)]_{j},
F_{140}^{n}(x)]_{j+1}\right).
\end{equation}
This will be called verification \emph{ab finito}.
Depending on the rule, either the first or the second method might be easier.
For rule 140, we will use \emph{ab finito} method.
Let $a=F_{140}^{n}(x)]_{j-1}$, $b=F_{140}^{n}(x)]_{j}$ and
$c=F_{140}^{n}(x)]_{j+1}$.
The right hand side of eq. (\ref{rule140abfinito}) is
$$
f_{140}\left(a,b,c\right)=
(1-a)b + abc, 
$$
where
\begin{align}
a&=\overline{x}_{j-2}x_{j-1}+\prod_{i=n-1}^{2n} x_{i-n+j-1},\\
b&=\overline{x}_{j-1}x_{j}+\prod_{i=n-1}^{2n} x_{i-n+j},\\
c&=\overline{x}_{j  }x_{j+1}+\prod_{i=n-1}^{2n} x_{i-n+j+1}.
\end{align}
Let us compute $(1-a)b$ first,
\begin{multline}
(1-a)b=\left(1- \overline{x}_{j-2}x_{j-1}-\prod_{i=n-1}^{2n} x_{i-n+j-1}\right)
\left( \overline{x}_{j-1}x_{j}+\prod_{i=n-1}^{2n} x_{i-n+j} \right)\\
=
\overline{x}_{j-1}x_{j}
+\prod_{i=n-1}^{2n} x_{i-n+j} 
-\overline{x}_{j-2}x_{j-1} \overline{x}_{j-1}x_{j} 
-\overline{x}_{j-2}x_{j-1} \prod_{i=n-1}^{2n} x_{i-n+j} \\
-\overline{x}_{j-1} x_j \prod_{i=n-1}^{2n} x_{i-n+j-1}
- \left(\prod_{i=n-1}^{2n} x_{i-n+j-1} \right) \left(\prod_{i=n-1}^{2n} x_{i-n+j} \right)
\end{multline}
The third term in the above, $\overline{x}_{j-2}x_{j-1} \overline{x}_{j-1}x_{j}=0$, 
because $x_{j-1} \overline{x}_{j-1}=0$. The fifth term vanishes because $\overline{x}_{j-1}$
multiplied by $x_{j-1}$ appearing in the product yields zero.  The last term
yields
\begin{multline*}
 \left(\prod_{i=n-1}^{2n} x_{i-n+j-1} \right) \left(\prod_{i=n-1}^{2n} x_{i-n+j} \right)
 =
 \left(\prod_{i=n-2}^{2n-1} x_{i-n+j} \right) \left(\prod_{i=n-1}^{2n} x_{i-n+j} \right)
 =\prod_{i=n-2}^{2n} x_{i-n+j}, 
\end{multline*}
where in the first equality we changed the dummy index in the first product from $i$ to $i+1$ and in the second equality we used the fact that
$x^2=x$ for Boolean variables.
This gives
\begin{multline}
(1-a)b=\overline{x}_{j-1}x_{j}
+\prod_{i=n-1}^{2n} x_{i-n+j}
-\overline{x}_{j-2} \prod_{i=n-1}^{2n} x_{i-n+j}
 -\prod_{i=n-2}^{2n} x_{i-n+j}\\
 =\overline{x}_{j-1}x_{j}
+\prod_{i=n-1}^{2n} x_{i-n+j}
-\overline{x}_{j-2} \prod_{i=n-1}^{2n} x_{i-n+j}
 -x_{j-2}\prod_{i=n-1}^{2n} x_{i-n+j}\\
 =\overline{x}_{j-1}x_{j}
+\prod_{i=n-1}^{2n} x_{i-n+j} - \prod_{i=n-1}^{2n} x_{i-n+j}=
\overline{x}_{j-1}x_{j}.
\end{multline}

Now we need to compute $abc$. We start from $ab$, which can be computed using
the result above, $ab=b-(1-a)b$, yielding
\begin{equation}
ab=\prod_{i=n-1}^{2n} x_{i-n+j}.
\end{equation}
 Now we multiply this by $c$,
\begin{multline}
abc= \left( \prod_{i=n-1}^{2n} x_{i-n+j} \right) 
\left( \overline{x}_{j  }x_{j+1}+\prod_{i=n-1}^{2n} x_{i-n+j+1} \right)\\
=\overline{x}_{j  }x_{j+1}  \prod_{i=n-1}^{2n} x_{i-n+j} 
+\left( \prod_{i=n-1}^{2n} x_{i-n+j} \right) 
\left(\prod_{i=n-1}^{2n} x_{i-n+j+1} \right)\\
=\left( \prod_{i=n-1}^{2n} x_{i-n+j} \right) 
\left(\prod_{i=n-1}^{2n} x_{i-n+j+1} \right)=
\left( \prod_{i=n-1}^{2n} x_{i-n+j} \right) 
\left(\prod_{i=n}^{2n+1} x_{i-n+j} \right)\\=
\prod_{i=n-1}^{2n+1} x_{i-n+j} .
\end{multline}
The final result is
\begin{equation}
(1-a)b + abc=\overline{x}_{j-1}x_{j} +\prod_{i=n-1}^{2n+1} x_{i-n+j},
\end{equation}
which is exactly the left hand side of eq. (\ref{rule140abfinito}), that is,
$[F_{140}^{n+1}(x)]_j$. This verifies that the solution of the initial value problem for rule 140 given by eq. (\ref{r140sol}) is indeed correct.
\section{Rule 156}
The local function of rule 156 is given by
\begin{multline}\label{local156}
f_{156}(x_1,x_2,x_3) = \overline{x_1}x_2+x_1x_2x_3 + x_1\overline{x}_2\overline{x}_3
=f_{140}(x_1,x_2,x_3) + x_1\overline{x}_2\overline{x}_3,
\end{multline}
meaning that  $f_{140}$ and $f_{155}$  differ only on the block $x_1x_2x_3=100$, otherwise they produce the same output.
This suggest that rule 156 can be viewed as a ``perturbed'' version of rule 140,
and that there is a chance that the formula for $[F^n_{156}(x)]_i$ may be somewhat
related to the formula for $[F^n_{140}(x)]_i$ derived in the previous section.
Figure~\ref{rule140-156} shows spatiotemporal patterns generated by
rules 140 (left) and  156 (right), starting from identical initial configurations. It is clear that the triangles with vertical strips below, present in the pattern of rule 140, also appear in the pattern of rule 156. Furthermore, it seems that every cell which is in state 1
in the pattern of rule 140 is also in state 1 in rule 156. 
\begin{figure}
\begin{center}
\includegraphics[scale=0.3]{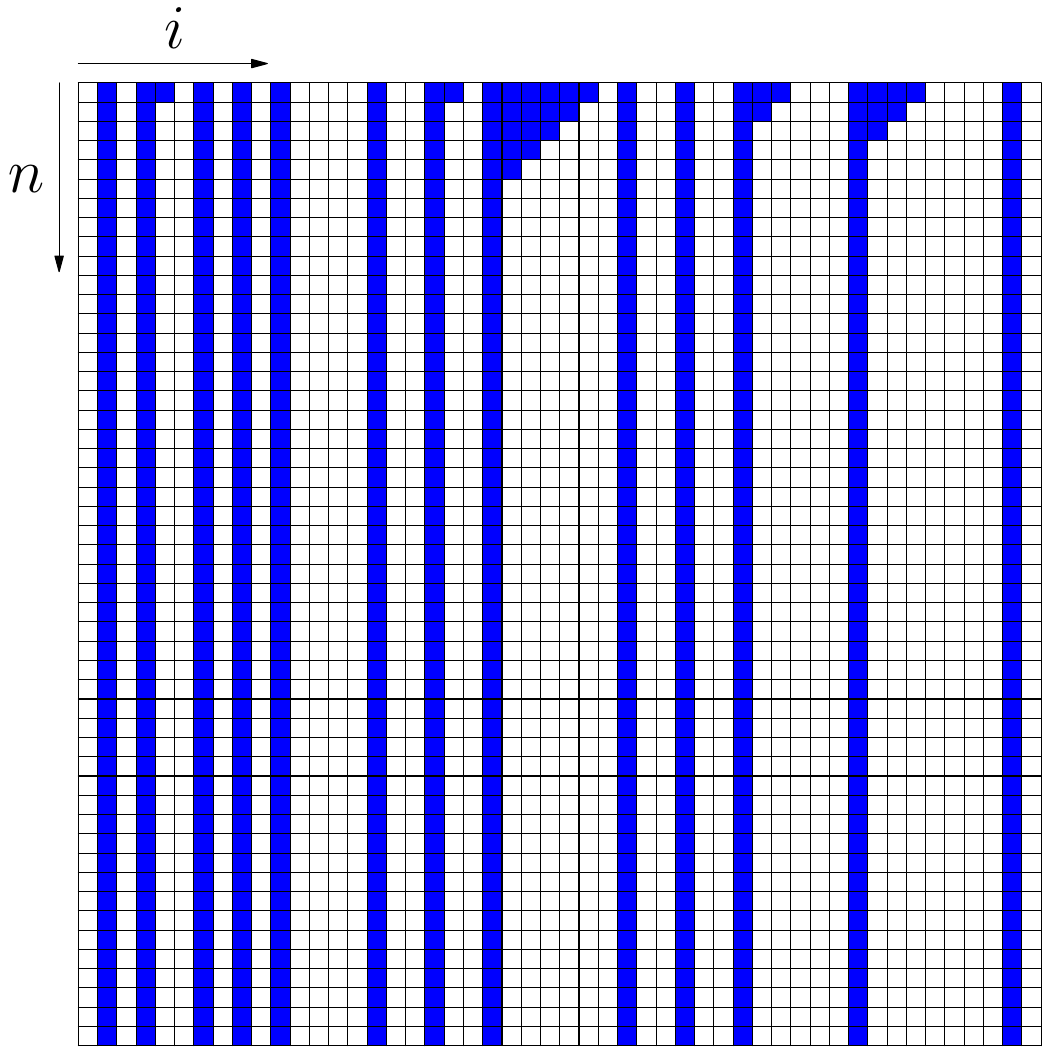}
\includegraphics[scale=0.3]{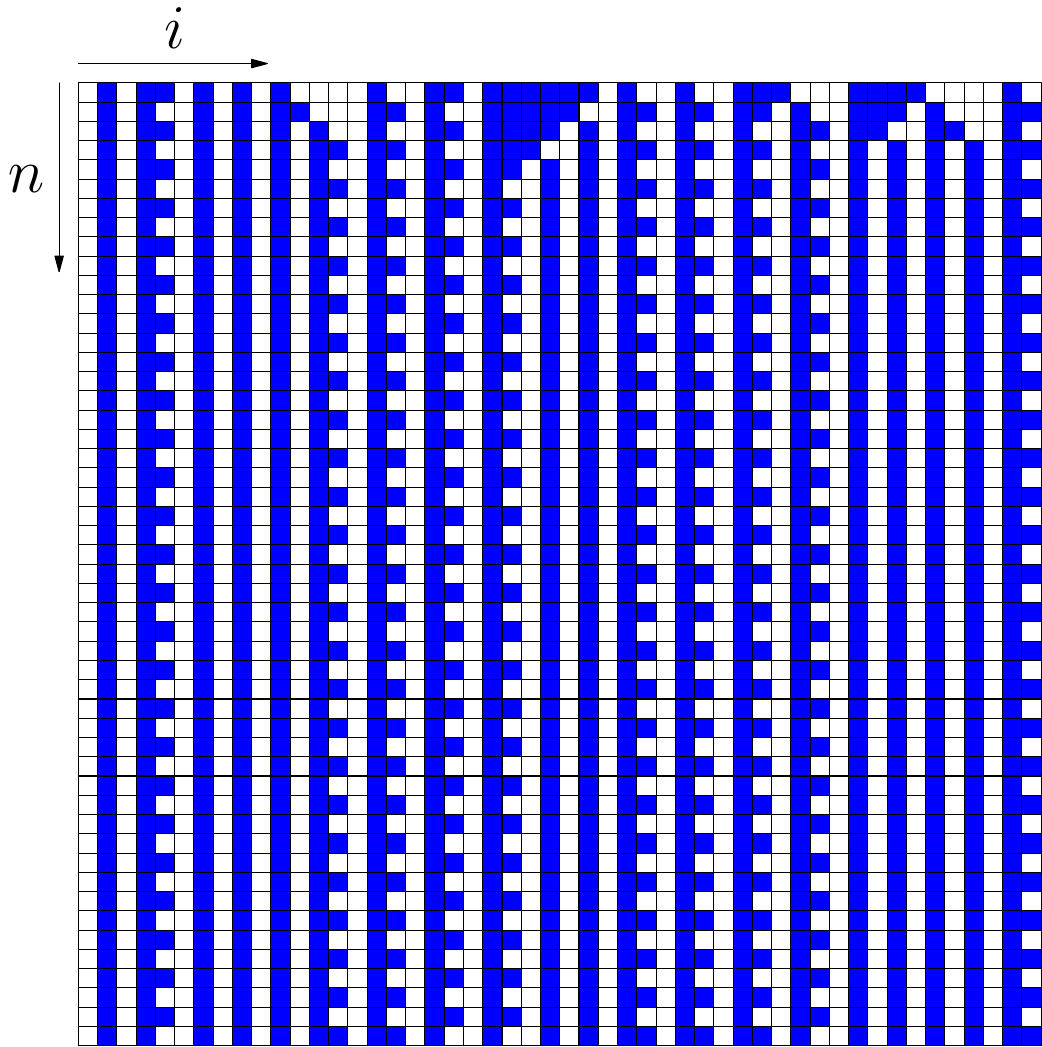}
\end{center}
\caption{Spatiotemporal patterns of rule 140 (left) and rule 156 (right).}
\label{rule140-156}
\end{figure}
This would mean that $[F^n_{156}(x)]_i \geq [F^n_{140}(x)]_i,$ or, in other words, 
$$
[F^n_{156}(x)]_i = [F^n_{140}(x)]_i + P(x,n,i),
$$
where $P(x,n,i)$ is some unknown non-negative function reflecting the
effects of the perturbation $x_1\overline{x}_2\overline{x}_3$.
\begin{figure}
\begin{center}
\includegraphics[scale=0.6]{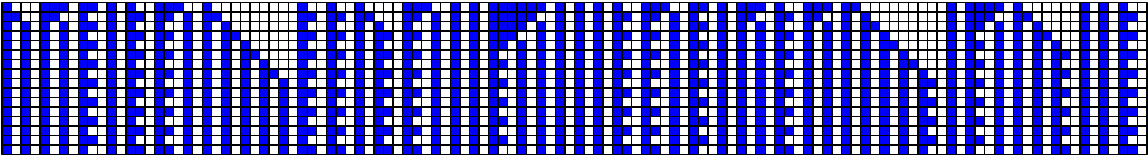}\\ $[F_{156}^n(x)]_i$\\[1em]
\includegraphics[scale=0.6]{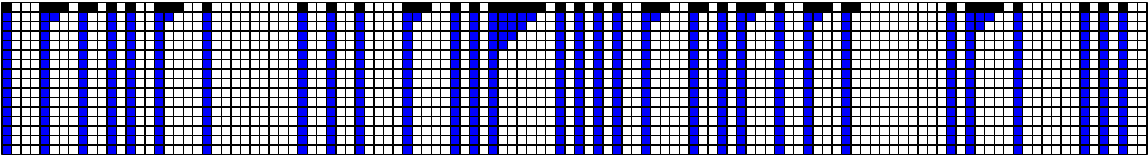}\\ $=[F_{140}^n(x)]_i$\\[1em]
\includegraphics[scale=0.6]{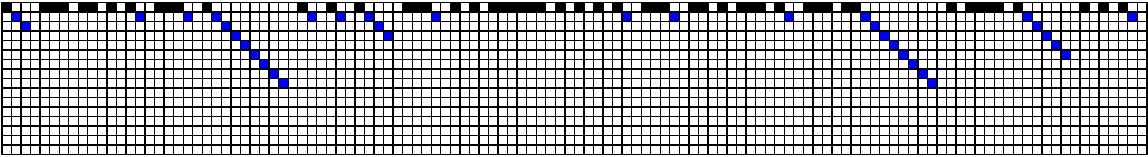}\\ $+D(x,n,i)$\\[1em]
\includegraphics[scale=0.6]{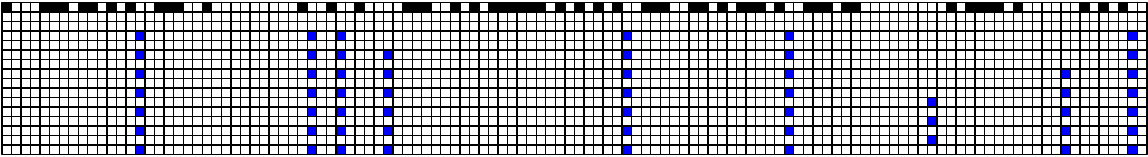}\\ $+B_1(x,n,i)$\\[1em]
\includegraphics[scale=0.6]{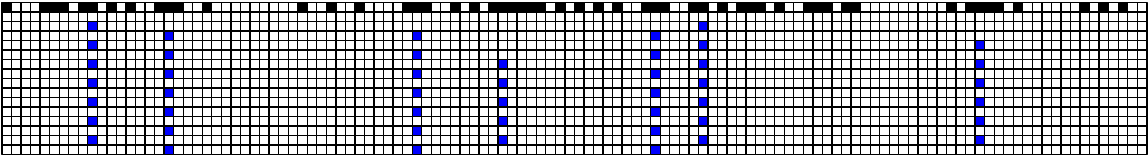}\\ $+B_2(x,n,i)$\\[1em]
\includegraphics[scale=0.6]{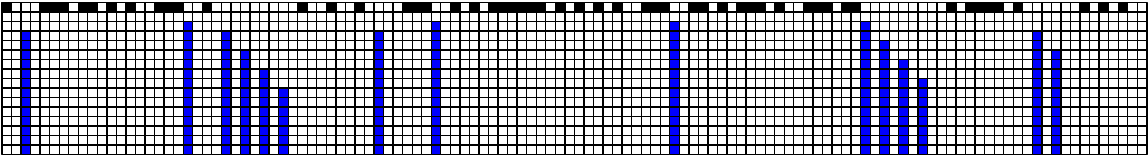}\\ $+S_1(x,n,i)$\\[1em]
\includegraphics[scale=0.6]{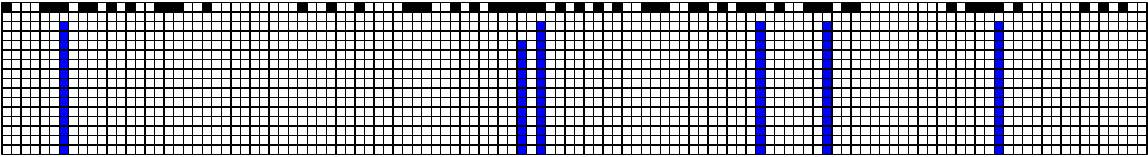}\\ $+S_2(x,n,i).$\\[1em]
\end{center}
\caption[Decomposition of the spatiotemporal pattern of rule 156.]{Spatiotemporal pattern of rule 156 (top) decomposed into six elements. Initial configuration is shown in black color in the decomposed elements.}
\label{rule156patt}
\end{figure}

We will now show how to construct the expression for the ``perturbation effect term'' $P(x,n,i)$ using a ``visual''
approach, that is, by analyzing the spatiotemporal pattern produced by rule 156 and by
decomposing is into simpler elements, each of them with relatively simple algebraic
description. The ``perturbation'' $P(x,n,i)$ appears in the spatiotemporal pattern of rule 156 as various additional elements not present in the pattern or rule 140, namely vertical strips, diagonal lines as well as ``blinkers'' (vertical lines of alternating 0's and 1's).  They are shown in Figure~\ref{rule156patt} marked with different labels, as follows:
\begin{itemize}
\item solid triangles and attached vertical strips as produced by rule 140;
\item diagonal lines, to be denoted by $D(x,n,i)$;
\item blinkers under diagonal lines, to be denoted by $B_1(x,n,i)$;
\item blinkers under solid triangles, to be denoted by $B_2(x,n,i)$;
\item vertical strips under diagonal lines, to be denoted by $S_1(x,n,i)$;
\item vertical strips under solid triangles, to be denoted by $S_2(x,n,i)$.
\end{itemize}
Cells in state 0 are shown as white and 1's in the initial configuration are black.
With the above labeling, the decomposition can now be formally stated as
\begin{multline}\label{rule156corrections2}
[F^n_{156}(x)]_i=[F^n_{140}(x)]_i \\+ D(x,n,i)
+ B_1(x,n,i)
+ B_2(x,n,i)
+ S_1(x,n,i)
+ S_2(x,n,i).
\end{multline}
The expression for $[F^n_{140}(x)]_i$ is given by eq. (\ref{r140sol}),
but we need to  construct the formulae corresponding to the five remaining terms
$D(x,n,i)$, $B_1(x,n,i)$, $ B_2(x,n,i)$, $ S_1(x,n,i)$ and $ S_2(x,n,i).$

We will not go into details of the construction of the relevant expressions for
all five terms, but since the construction is quite similar for all of them,
we will show how to do it for the first two.  

 Let us start from the diagonal line $D(x,n,i)$. Figure~\ref{rule156patt} suggests that the cell $j$ after $n$ iterations will belong to the diagonal  line if  $x_{j-n}=1$
and if it lies below the cluster of continuous zeros, so that 
$x_{m}=0$ for all $m \in \{j-n+1,  \ldots, j+1 \}$. The last condition will be realized if the product $\prod _{m=j-n+1}^{j+1}\overline{x}_{m}$
is equal to 1,  yielding the expression
\begin{equation}\label{correctionsbegin}
D(x,n,j)=x_{j-n}\prod _{m=j-n+1}^{j+1}\overline{x}_{m}.
\end{equation}

Blinkers are the next. Blinkers of  $B_1(x,n,i)$ type lie below the diagonal lines.
\begin{figure}
\begin{center}
\includegraphics[width=11cm]{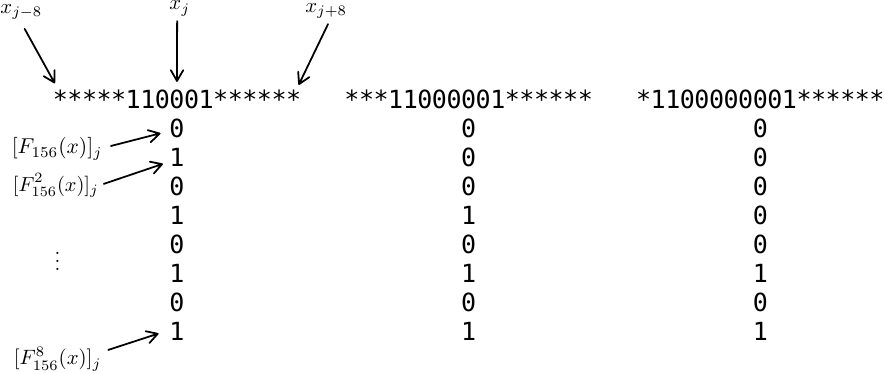}
\end{center}
\caption{Three initial configurations producing blinkers of even type.}\label{blinkfig}
\end{figure}
They occur below clusters of zeros in the initial configuration,
and can be in state 1 on even or odd $n$ values. Let us consider first those which
are 1 on even values, to be denoted by $B_1^{ev}$. Careful analysis of the 
spatiotemporal patterns reveals that they need to lie below clusters of zeros of odd length
preceded by 11 and terminated by 1. For the sake of example consider $n=8$.
In this case $B_1^{ev}(x,8,j)$ can possibly depend  on  8 left and 8 right 
neighbours of $x_j$. The blinker lies below the penultimate zero in the cluster of 0's,
therefore the cluster can only be of length 3, 5, or 7. Figure \ref{blinkersfig}
shows the initial configuration of $x_j$ with 8 left and 8 right neighbours for all three cases,
together with the resulting column below. Stars denote arbitrary symbols. 

We can see that blinking can start sooner or later depending on the 
length of the cluster  of zeros. This is because blinkers lie below diagonals. The three configurations shown above
correspond to
\begin{align*}
(x_{j-3},x_{j-2},\ldots, x_2)&= (1,1,0,0,0,1),\\
(x_{j-5},x_{j-2},\ldots, x_2) &=(1,1,0,0,0,0,0,1),\\
(x_{j-7},x_{j-2},\ldots, x_2) &= (1,1,0,0,0,0,0,0,0,1).
\end{align*}
This means that  
\begin{align*}B^{ev}_1(x,8,j)=&x_{j-3}x_{j-2}\overline{x}_{j-1}\overline{x}_{j}\overline{x}_{j+1}x_{j+2}\\
+&x_{j-5}x_{j-4}\overline{x}_{j-3}\overline{x}_{j-2}\overline{x}_{j-1}\overline{x}_{j}\overline{x}_{j+1}x_{j+2}\\
+&x_{j-7}x_{j-6}\overline{x}_{j-5}\overline{x}_{j-4}\overline{x}_{j-3}\overline{x}_{j-2}\overline{x}_{j-1}\overline{x}_{j}\overline{x}_{j+1}x_{j+2}.
\end{align*}
Defining
$$
\mathcal{P}(n)=
\begin{cases}
  1  & n \text{ is even,} \\
  0 & n \text{ is odd,}
\end{cases}
$$
this can be written as
$$B^{ev}_1(x,8,j)= 
 \sum _{r=2}^{8} \left( \mathcal{P
} \left( r+1 \right) x_{{j-r}}x_{{j-r+1}} x_{{j+2}} \prod _{m=j-r+2}^{j+1}\overline{x}_{{m
}} \right).
$$
It is now straightforward to guess the general formula for blinker which is 1 on even $n$,
\begin{equation*}
 B_1^{ev}(x,n,j)=\mathcal{P} \left( n \right) \sum _{r=2}^{n} \left( \mathcal{P} \left( r+1 \right) x_{{j-r}}x_{{j-r+1}} x_{{j+2}} \prod _{m=j-r+2}^{j+1}\overline{x}_{{m
}} \right). 
\end{equation*}
For blinkers which are 1 on odd $n$ the analysis is very similar, except that they 
lie below clusters of even number of zeros preceded by 01. Denoting them by $B_1^{odd}$,
this yields
\begin{equation*}
 B_1^{odd}(x,n,j)=\mathcal{P} 
\left( n+1 \right) 
\sum _{r=2}^{n}
 \left( \mathcal{P} \left( r \right) 
\overline{x}_{{j-r}} 
   x_{{j-r
+1}} x_{{j+2}} \prod _{m=j-r+2}^{j+1}\overline{x}_{{m}} \right). 
\end{equation*}
The final expression for blinkers $B_1$ is the sum of $B_1^{ev}$ and $B_1^{odd}$,
\begin{multline}
 B_1(x,n,j)=\mathcal{P} \left( n \right) \sum _{r=2}^{n} \left( \mathcal{P} \left( r+1 \right) x_{{j-r}}x_{{j-r+1}} x_{{j+2}} \prod _{m=j-r+2}^{j+1}\overline{x}_{{m
}} \right) 
\\+\mathcal{P} 
\left( n+1 \right) 
\sum _{r=2}^{n}
 \left( \mathcal{P} \left( r \right) 
\overline{x}_{{j-r}} 
   x_{{j-r
+1}} x_{{j+2}} \prod _{m=j-r+2}^{j+1}\overline{x}_{{m}} \right). 
\end{multline}

The other blinkers, $B_2(x,n,j)$, are very similar, except that they
are under clusters of 1's preceded by 0 and terminated by either 00 or 01,
depending on the parity of $n$ for which they are in state 1. The corresponding expression is then straightforward to
construct,
\begin{multline}
 B_2(x,n,j)=\mathcal{P} \left( n+1 \right) \sum _{m=0}^{n-1} \left(
 \mathcal{P} \left( m+1 \right)  
\overline{x}_{{j-2}}
 \overline{x}_{{j+m+1}} 
 \overline{x}_{{j+m+2}} 
\prod _{k=j-1}^{j+m}x_{{k}} \right)\\ 
 +\mathcal{P} \left( n \right) \sum _{m=0}^{n-1}
 \left( \mathcal{P} \left( m \right) 
   \overline{x}_{{j-2}}
\overline{x}_{{j+m+1}} 
 x_{{j+m+2}} \prod _{k=j-1}^{j+m}x_{k}  \right). 
\end{multline}
The final two items we need to consider are the two types of vertical strips. The first ones are strips under clusters of 1's starting with 11 or 01. Using similar reasoning as for the blinkers above, this corresponds to the expression
\begin{multline}
 S_1(x,n,j)=\sum _{k=1}^{n} \left( \mathcal{P} \left( k \right) x_{{j-k}
}x_{{j-k+1}}\prod _{m=j-k+2}^{1+j}\overline{x}_{{m}} \right)\\ +\sum _{k=2}^{n}
 \left( \mathcal{P} \left( k+1 \right) 
  \overline{x}_{{j-k}} 
   x_{{j
-k+1}}\prod _{m=j-k+2}^{1+j}\overline{x}_{{m}} \right) .
\end{multline}
Here we again deal with structures occurring below 
clusters of 0's, hence the products $\prod _{m=j-k+2}^{1+j}\overline{x}_{{m}}$. There are two sums because
there are two different expressions depending on 
the parity, and $\mathcal{P} ( k )$
and $\mathcal{P} ( k +1)$ take care of this.
The other type of strips,  $S^{(1)}_{n,j}$, corresponds to analogous expression,
\begin{multline}\label{correctionsend}
S_2(x,n,j)=\sum _{k=2}^{n} \left(  
\mathcal{P}(k) x_{{j+k}} \overline{x}_{{j-1+k}}  
  \prod _{m=j-2}^{j+k-2}x_{{m}} \right)
  +\\
  \sum _{k=2}^{n} \left(  
\mathcal{P}(k+1) \overline{x}_{{j+k}} \overline{x}_{{j-1+k}} 
  \prod _{m=j-2}^{j+k-2}x_{{m}} \right). 
\end{multline}

Eq.  (\ref{rule156corrections2}) together with expressions
defined in eqs. (\ref{correctionsbegin}--\ref{correctionsend}) provide a complete solution of the initial value problem for rule 156.
\begin{theorem}\label{maintheorem156}
For elementary cellular automaton rule 156, 
for any $x\in \{0,1\}^{\mathbb{Z}}$ and $n>1$, the state of the $j$-th cell after $n$ iterations of the rule starting from $x$ is given by
\begin{multline}\label{rule156solution}
[F^n_{156}(x)]_j=
\overline{x}_{j-1}x_{j}+\prod_{i=n-1}^{2n} x_{i-n+j}
 \\+ D(x,n,j)
+ B_1(x,n,j)
+ B_2(x,n,j)
+ S_1(x,n,j)
+ S_2(x,n,j),
\end{multline}
where functions $D$, $B_1$, $B_2$, $S_1$ and $S_2$ are defined by
\begin{align*}
D(x,n,j)&=x_{j-n}\prod _{m=j-n+1}^{j+1}\overline{x}_{m},\\
 B_1(x,n,j)&=\mathcal{P} \left( n \right) \sum _{r=2}^{n} \left( \mathcal{P
} \left( r+1 \right) x_{{j-r}}x_{{j-r+1}} x_{{j+2}} \prod _{m=j-r+2}^{j+1}\overline{x}_{{m
}} \right) \\
&\qquad \qquad +\mathcal{P} 
\left( n+1 \right) 
\sum _{r=2}^{n}
 \left( \mathcal{P} \left( r \right) 
\overline{x}_{{j-r}} 
   x_{{j-r
+1}} x_{{j+2}} \prod _{m=j-r+2}^{j+1}\overline{x}_{{m}} \right), \\
 B_2(x,n,j)&=\mathcal{P} \left( n+1 \right) \sum _{m=0}^{n-1} \left(
 \mathcal{P} \left( m+1 \right)  
\overline{x}_{{j-2}}
 \overline{x}_{{j+m+1}} 
 \overline{x}_{{j+m+2}} 
\prod _{k=j-1}^{j+m}x_{{k}} \right)\\ 
 &\qquad \qquad+\mathcal{P} \left( n \right) \sum _{m=0}^{n-1}
 \left( \mathcal{P} \left( m \right) 
   \overline{x}_{{j-2}}
\overline{x}_{{j+m+1}} 
 x_{{j+m+2}} \prod _{k=j-1}^{j+m}x_{k}  \right), \\
 S_1(x,n,j)&=\sum _{k=1}^{n} \left( \mathcal{P} \left( k \right) x_{{j-k}
}x_{{j-k+1}}\prod _{m=j-k+2}^{1+j}\overline{x}_{{m}} \right)\\ 
&\qquad \qquad+\sum _{k=2}^{n}
 \left( \mathcal{P} \left( k+1 \right) 
  \overline{x}_{{j-k}} 
   x_{{j
-k+1}}\prod _{m=j-k+2}^{1+j}\overline{x}_{{m}} \right),\\
S_2(x,n,j)&=\sum _{k=2}^{n} \left(  
\mathcal{P}(k) x_{{j+k}} \overline{x}_{{j-1+k}}  
  \prod _{m=j-2}^{j+k-2}x_{{m}} \right)\\
  &\qquad \qquad+
  \sum _{k=2}^{n} \left(  
\mathcal{P}(k+1) \overline{x}_{{j+k}} \overline{x}_{{j-1+k}} 
  \prod _{m=j-2}^{j+k-2}x_{{m}} \right). 
\end{align*}
\end{theorem}

 A formal proof of the correctness of the above solution
formula can be obtained by the \emph{ab initio} method described earlier, that is, by verifying that for any $j \in \mathbb{Z}$, $n>0$ and $x\in \{0,1\}^{\mathbb{Z}}$,
\begin{equation} \label{proofformula}
[F_{156}^{n+1}(x)]_j=[F_{156}^{n}(y)]_j,
\end{equation}
where $y_i=f_{156}(x_{i-1},x_i,x_{i+1})$. The proof will be presented in the next section.
\section{Proof of  correctness of the  solution formula}
On the right hand side of the equality we want to prove, eq. (\ref{proofformula}), we have
 $[F_{156}^{n}(y)]_j$. If we expand this using the solution formula, we will
 obtain a number of terms involving various products of $y$ variables. We need to prove two lemmas which will help to simplify these products.
\begin{lemma}\label{lemma1}
Let $f_{156}$ be the local function of the ECA 156 given by eq. (\ref{local156})
and let $x_{-1},x_0, x_1, \ldots x_{n+1}$  be Boolean variables, where $n>1$. If $y_i=f_{156}(x_{i-1},x_i,x_{i+1})$ for all $i\in \{0, 1,\ldots,n\}$, then
\begin{equation}\label{prodlemma}
(i) \,\,\,\,\,\,\,\, \prod_{i=0}^n y_i = \prod_{i=0}^{n+1} x_i, 
\end{equation}
and
\begin{equation}\label{prodlemmaii}
(ii) \,\,\,\,\,\,\,\, \prod_{i=0}^n \overline{y}_i = \prod_{i=-1}^{n} \overline{x}_i. 
\end{equation}
\end{lemma}
\begin{proof}
We will prove (i) by induction. 
Recall that $f_{156}(x_1,x_2,x_3) = x_1x_2x_3 + x_1\overline{x}_2\overline{x}_3 +\overline{x_1}x_2$,
thus  for $n=2$ we have
\begin{multline}\prod_{i=0}^2 y_i =y_0y_1y_2 =
f_{156}(x_{-1},x_0,x_1)
f_{156}(x_{0},x_1,x_2)
f_{156}(x_{1},x_2,x_3)\\=
 \left( x_{-1}x_{0}x_{1}+x_{-1}\overline{x}_{0}\overline{x}_{1}+x_{0}\overline{x}_{-1}
 \right) \\ \times \left( x_{0}x_{1}x_{2}+x_{0}\overline{x}_{1}\overline{x}_{2}+x_{1}\overline{x}_
{0} \right)  \left( x_{1}x_{2}x_{3}+x_{1}\overline{x}_{2}\overline{x}_{3}+x_{
2}\overline{x}_{1} \right). 
\end{multline}
Expanding the above expression and simplifying it using $x_i^p=x_i$, $\overline{x}_i^p=\overline{x}_i$ and $x_i+\overline{x}_i=1$,
one obtains
$$y_0y_1y_2 = x_0 x_1 x_2 x_3,$$
confirming correctness of eq. (\ref{prodlemma}) for $n=2$.

Now let us suppose that eq. (\ref{prodlemma}) is valid for a given $n$. We have
\begin{multline}
\prod_{i=0}^{n+1} y_i=
y_{n+1}\prod_{i=0}^{n} y_i=y_{n+1}\prod_{i=0}^{n+1} x_i=
\left(x_{n}x_{n+1}x_{n+2}+x_{n}\overline{x}_{n+1}\overline{x}_{n+2}+x_{n+1}\overline{x}_{n} \right) \prod_{i=0}^{n+1} x_i\\
= x_{n}x_{n+1}x_{n+2} \prod_{i=0}^{n+1} x_i+
\underbrace{x_{n}\overline{x}_{n+1}\overline{x}_{n+2} \prod_{i=0}^{n+1} x_i}_{=0}+
\underbrace{x_{n+1}\overline{x}_{n} \prod_{i=0}^{n+1} x_i}_{=0}
=\prod_{i=0}^{n+2} x_i,
\end{multline}
where we used the properties of Boolean variables $x_{n+1}\overline{x}_{n+1}=0$ and $x_{n}\overline{x}_{n}=0$ as well as
$x_{n+1}^2=x_{n+1}$ and $x_{n}^2=x_{n}$. This verifies the validity of  eq. (\ref{prodlemma}) for $n+1$, completing the 
induction step. 

Proof of (ii) is very similar. First we will note that
 $f_{156}(x_1,x_2,x_3)$ can be explicitly expressed as
$$f_{156}(x_1,x_2,x_3) = \begin{cases}
  1  &  \text{ if $x_1x_2x_3=010, 011, 100$ or  $111$, } \\
  0 & \text{ otherwise.}
\end{cases}$$
From this it follows that
$$1-f_{156}(x_1,x_2,x_3) = \begin{cases}
  1  &  \text{ if $x_1x_2x_3=000, 001, 101$ or  $110$, } \\
  0 & \text{ otherwise,}
\end{cases}$$
thus we can write
\begin{multline}\label{oneminusy}
1-f_{156}(x_1,x_2,x_3)=\overline{x}_1\overline{x}_2\overline{x_3}+\overline{x}_1\overline{x}_2x_3
+x_1\overline{x}_2x_3 + x_1x_2\overline{x}_3\\
=\overline{x}_1\overline{x}_2+x_1\overline{x}_2x_3 + x_1x_2\overline{x}_3.
\end{multline}
For $n=2$ the left hand side of eq. (\ref{prodlemmaii})
becomes
\begin{multline}\prod_{i=0}^2 \overline{y}_i =\overline{y}_0\overline{y}_1\overline{y}_2 =
(1-f_{156}(x_{-1},x_0,x_1))
(1-f_{156}(x_{0},x_1,x_2))
(1-f_{156}(x_{1},x_2,x_3))\\=
 \left( \overline{x}_{-1}\overline{x}_0+x_{-1}\overline{x}_0x_1 + x_{-1}x_0\overline{x}_1  \right) \\ 
 \times 
 \left(\overline{x}_0\overline{x}_1+x_0\overline{x}_1x_2 + x_0x_1\overline{x}_2 \right) 
 \left( \overline{x}_1\overline{x}_2+x_1\overline{x}_2x_3 + x_1x_2\overline{x}_3 \right). 
\end{multline}
Expanding the above expression and simplifying it in a similar fashion as before we obtain
$$\overline{y}_0\overline{y}_1\overline{y}_2=\overline{x}_{-1} \overline{x}_{0} \overline{x}_{1} \overline{x}_{2},$$
confirming correctness of eq. (\ref{prodlemmaii}) for $n=2$.

For the induction step, assume eq.  eq. (\ref{prodlemmaii}) is valid for $n$.
Then we have
\begin{multline}
\prod_{i=0}^{n+1} \overline{y}_i=
\overline{y}_{n+1}\prod_{i=0}^{n} \overline{y}_i=
\overline{y}_{n+1}\prod_{i=-1}^{n} \overline{x}_i=
\left(  \overline{x}_{n}\overline{x}_{n+1} + x_{n}x_{n+1}\overline{x}_{n+2}+x_{n}\overline{x}_{n+1} x_{n+2} \right) \prod_{i=-1}^{n} \overline{x}_i\\
=  \overline{x}_{n}\overline{x}_{n+1}  \prod_{i=-1}^{n} \overline{x}_i
+  \underbrace{x_{n}x_{n+1}\overline{x}_{n+2}  \prod_{i=-1}^{n} \overline{x}_i}_{=0}
+   \underbrace{x_{n}\overline{x}_{n+1} x_{n+2}  \prod_{i=-1}^{n} \overline{x}_i}_{=0}
=\prod_{i=0}^{n+1} \overline{x}_i,
\end{multline}
This confirms the validity of eq. (\ref{prodlemmaii}) for $n+1$, completing the induction step. \end{proof}

Lemma \ref{lemma1} deals with products of more than two consecutive $y$ variables, but
we will also need to simplify products of two variables. The next result
provides for this.
\begin{lemma}\label{lemma2}
If $x\in\{0,1\}^{\mathbb{Z}}$ and $y=F_{156}(x)$ then
\begin{align*}
y_k=& x_{{k}}x_{{k-1}}x_{{k+1}}+x_{{k-1}}\overline{x}_{{k}}\overline{x}_{{k+1}}+x_{{k}}\overline{x}_{{k-1}},   & (i)\\
\overline{y}_k=&  x_{{k}}x_{{k-1}}\overline{x}_{{k+1}}+x_{{k-1}}\overline{x}_{{k}}x_{{k+1}}+\overline{x}_{{k}}\overline{x}_{{k-1}}, & (ii)\\
y_{k} y_{k+1}=&x_k x_{k+1} x_{k+2} + \overline{x}_{k-1} x_k \overline{x}_{k+1} \overline{x}_{k+2}, & (iii)\\
\overline{y}_{k} y_{k+1}=&x_k \overline{x}_{k+1} 
   + x_{k-1} x_k \overline{x}_{k+1} \overline{x}_{k+2}, & (iv)\\
\overline{y}_{k} \overline{y}_{k+1}=&   \overline{x}_{k-1} \overline{x}_k \overline{x}_{k+1}
+ x_{k-1}x_k \overline{x}_{k+1} x_{k+2}. & (v)
\end{align*}
\end{lemma}
\begin{proof}
Formula (i) is a direct consequence of the definition of $f_{156}$.
Expression (ii) has already been derived in eq. (\ref{oneminusy}).
To prove (iii-v), we write the relevant product
 using (i-ii) and then expand and simplify the
resulting expression using properties of Boolean variables. For example, for (iv)
we have
\begin{multline}
\overline{y}_{k}y_{k+1}=\\
 \left( x_{k+1}x_{k-1}\overline{x}_{k}+x_{k-1}x_{k-1}\overline{x}_{k+1}+\overline{x}_{k-1}
 \overline{x}_{k} \right)  \left( x_{k+1}x_{k}x_{k+2}+x_{k}\overline{x}_{k+1}\overline{x}_
{k+2}+x_{k}\overline{x}_{k+1} \right). 
\end{multline}
When the above product is expanded, one obtains 9 terms, but most of them 
are equal to zero because they contain product of complementary Boolean variables.
Only two terms remain in the end, yielding
\begin{equation}
\overline{y}_{k}y_{k+1}
=x_{{k-1}}x_{{k}}\overline{x}_{{k+1}}\overline{x}_{{k+2}}+x_{{k}}\overline{x}_{{k+1}}.
\end{equation}
Proofs of (iii) and (v) can be procured in a similar fashion.
\end{proof}
The next lemma describes relationship between quantities defined in 
eqs. (\ref{correctionsbegin}--\ref{correctionsend}) in $x$ and $y$ variables.
\begin{lemma}\label{lemma3}
If $x\in \{0,1\}^\mathbb{Z}$ and $y=F_{156}(x)$, then 
the following identities are satisfied:
\begin{align*}
[F^n_{140}(y)]_j&= [F^{n+1}_{140}(x)]_j+ x_{j-2}x_{j-1}\overline{x}_{j}\overline{x}_{j+1}, & (i)\\
D(y,n,j)&=D(x,n+1,j), & (ii)\\
B_1(y,n,j)+B_2(y,n,j)&=B_1(x,n+1,j)+B_1(x,n+1,j), & (iii)\\
S_1(y,n,j)+S_2(y,n,j)&=S_1(x,n+1,j)+S_2(x,n+1,j) - x_{j-2}x_{j-1}\overline{x}_{j}
\overline{x}_{j+1}. & (iv)
\end{align*}
\end{lemma}
\begin{proof}
(i) Solution of the initial value problem for rule 140 given by eq. (\ref{r140sol}) yields
\begin{equation*}
[F^n_{140}(y)]_j=
\overline{y}_{j-1}y_{j}+\prod_{i=n-1}^{2n} y_{i-n+j}.
\end{equation*}
Using Lemma \ref{lemma2}(iv), the first term of the above becomes
\begin{equation*}
\overline{y}_{j-1}y_{j}=x_{{j-2}}x_{{j-1}}\overline{x}_{{j}}\overline{x}_{{j+1}}+x_{{j}}\overline{x}_{{j-1}}.
\end{equation*}
The second term, by the virtue of Lemma \ref{lemma1}(i),
is given by
\begin{equation*}
\prod_{i=n-1}^{2n} y_{i-n+j}=
\prod_{i=n-1}^{2n+1} x_{i-n+j}=\prod_{p=n}^{2n+2} x_{p-n-1+j}.
\end{equation*}
The last equality reflect the change of the dummy index from $i$ to $p$,  $i=p-1$. 
The final result is then
\begin{multline*}
[F^n_{140}(y)]_j=x_{{j-2}}x_{{j-1}}\overline{x}_{{j}}\overline{x}_{{j+1}}+x_{{j}}\overline{x}_{{j-1}}
+\prod_{p=n}^{2n+2} x_{p-n-1+j}\\
=x_{j-2}x_{j-1}\overline{x}_{j} \overline{x}_{j+1}+[F^{n+1}_{140}(x)]_j, 
\end{multline*}
as required.

(ii) The next identity we will prove involves $D$. Using Lemma \ref{lemma1}(i) as well as Lemma \ref{lemma2}(ii), we have
\begin{multline}
D(y,n,j)=y_{j-n}\prod _{m=j-n+1}^{j+1}\overline{y}_{m}\\
=\left(
x_{{j-n}}x_{{j-n-1}}x_{{j-n+1}}+x_{{j-n-1}}\overline{x}_{{j-n}}\overline{x}_{{j-n+1}}+x_{{j-
n}}\overline{x}_{{j-n-1}}
\right)
\prod_{j-n}^{j+1}\overline{x}_m\\
=x_{j-n-1}  \prod_{m=j-n}^{j+1}\overline{x}_m
=x_{j-(n+1)}  \prod_{m=j-(n+1)+1}^{j+1}\overline{x}_m
=D(x,n+1,j).
\end{multline}

(iii) Let us deal with $B_1$ first. Let us assume that $n$ is even, then
\begin{equation*}
 B_1(y,n,j)= \sum _{r=2}^{n} \left( \mathcal{P
} \left( r+1 \right) y_{{j-r}}y_{{j-r+1}} y_{{j+2}} \prod _{m=j-r+2}^{j+1}\overline{y}_{{m
}} \right). 
\end{equation*}
The expression inside the sum can be transformed as follows,
\begin{multline*}
y_{{j-r}}y_{{j-r+1}} y_{{j+2}} \prod _{m=j-r+2}^{j+1}\overline{y}_{{m
}} \\
=\left(
x_{j-r} x_{j-r+1} x_{j-r+2} + \overline{x}_{j-r-1} x_{j-r} \overline{x}_{j-r+1} \overline{x}_{j-r+2}\right)y_{{j+2}} \prod _{m=j-r+1}^{j+1}\overline{x}_{m}
 \\
= \overline{x}_{j-r-1} x_{j-r}  
y_{{j+2}} \prod _{m=j-r+1}^{j+1}\overline{x}_{m}
\\=
 \overline{x}_{j-r-1} x_{j-r}  
\left(
x_{j+2}x_{j+1}x_{j+3}+x_{j+1}\overline{x}_{j+2}\overline{x}_{j+3}+x_{j+2}\overline{x}_{j+1}
\right)
 \prod _{m=j-r+1}^{j+1}\overline{x}_{m}\\
 =
 \overline{x}_{j-r-1} x_{j-r}  
x_{j+2}
 \prod _{m=j-r+1}^{j+1}\overline{x}_{m},
\end{multline*}
If $n$ is even, we thus obtain
\begin{equation}B_1(y,n,j)=\sum_{r=2}^n \left( \mathcal{P}(r+1)   \overline{x}_{j-r-1} x_{j-r}  
x_{j+2}
 \prod _{m=j-r+1}^{j+1}\overline{x}_{m} \right).
 \end{equation}
For even $n$, $n+1$ is odd, and from the definition of $B_1$,
 \begin{equation}
 B_1(x,n+1,j)=
\sum _{r=2}^{n+1}
 \left( \mathcal{P} \left( r \right) 
\overline{x}_{{j-r}} 
   x_{{j-r
+1}} x_{{j+2}} \prod _{m=j-r+2}^{j+1}\overline{x}_{{m}} \right). 
\end{equation}
Changing the index in the last sum to $r=p+1$ we get
 \begin{multline}
 B_1(x,n+1,j)=
\sum _{p=1}^{n}
 \left( \mathcal{P} \left( p+1 \right) 
\overline{x}_{{j-p-1}} 
   x_{{j-p}} x_{{j+2}} \prod _{m=j-p+1}^{j+1}\overline{x}_{{m}} \right)\\
   = \overline{x}_{{j-2}} 
   x_{{j-1}} x_{{j+2}} \prod _{m=j}^{j+1}\overline{x}_{{m}} 
   +B_1(y,n,j),
\end{multline}
thus
\begin{equation*}
B_1(y,n,j)=
B_1(x,n+1,j)- \overline{x}_{{j-2}} 
   x_{{j-1}} x_{{j+2}} \prod _{m=j}^{j+1}\overline{x}_{{m}}.
\end{equation*}
Using the same method, one can show that a similar identity holds for $B_2$, namely
\begin{equation*}
B_2(y,n,j)=
B_2(x,n+1,j)+ \overline{x}_{{j-2}} 
   x_{{j-1}} x_{{j+2}} \prod _{m=j}^{j+1}\overline{x}_{{m}}.
\end{equation*}
Formula (iii)  $n$ then follows automatically. For odd $n$ the reasoning is 
very similar, thus we will not be repeating it here.

(iv) For the last identity we will not supply all details, as the 
calculations are analogous to those performed in the proof of (iii).
By expressing $S_1(y,n,j)$ and $S_1(y,n,j)$ in terms of $x$ one obtains
$$S_1(y,n,j)=S_1(x,n+1,j)
+x_{j-2} x_{j-1} x_j \overline{x}_{j+1}  x_{j+2} -
x_{j-2} x_{j-1} \overline{x}_j \overline{x}_{j+1},  
$$
$$S_2(y,n,j)=S_2(x,n+1,j)-x_{j-2} x_{j-1} x_{j} \overline{x}_{j+1} x_{j+2}.$$
The above two equations, when added side by side, yield the identity (iv).
\end{proof}

We are now ready to prove Theorem \ref{maintheorem156}.
As remarked at the end of the previous section, all we need to do is to verify
eq. (\ref{proofformula}),
\begin{equation} \label{proofformulabis}
[F_{156}^{n+1}(x)]_j=[F_{156}^{n}(y)]_j.
\end{equation}
It should be obvious by now that Lemma \ref{lemma3} supplies all necessary 
identities. The right hand side of eq. (\ref{proofformulabis}) is
\begin{multline}\label{termsofsol}
[F^n_{156}(y)]_j=
[F^n_{140}(y)]_j
 \\+ D(y,n,j)
+ B_1(y,n,j)
+ B_2(y,n,j)
+ S_1(y,n,j)
+ S_2(y,n,j).
\end{multline}
By the virtue of Lemma \ref{lemma3} and because of the cancellation of the
term $x_{j-2}x_{j-1}\overline{x}_{j}\overline{x}_{j+1}$
 this becomes
\begin{multline}
[F^n_{156}(y)]_j=
[F^{n+1}_{140}(x)]_j
 + D(x,n+1,j)
+ B_1(x,n+1,j)\\
+ B_2(x,n+1,j)
+ S_1(x,n+1,j)
+ S_2(x,n+1,j),
\end{multline}
which is precisely the left hand side of  eq. (\ref{proofformulabis}),
concluding the proof of Theorem \ref{maintheorem156}.

\section{Probabilistic solution}
One of the useful applications of the explicit solution of the IVP is in determining the
``density'' of cells in state 1 after $n$ iterations of the rule. Suppose that
all values of $x_i$ are initially set as independent and identically distributed random variables, such that $Pr(x_i=1)=p$ and $Pr(x_i=0)=1-p$, where $p\in[0,1]$ is a fixed
parameter. The expected value of $x_i$ then corresponds to the ``density of ones''
in the initial configuration, and we have
$\langle x_i \rangle=1\cdot p + 0 \cdot (1-p)=p$. The initial density is thus $p$.
The density after $n$ iterations will be equal to $\langle [F^n(x)]_i\rangle$ and we need to compute it. Eq. (\ref{rule156corrections2}) yields
\begin{multline}
\langle [F^n_{156}(x)]_i \rangle= \langle[F^n_{140}(x)]_i \rangle + \langle D(x,n,i)\rangle
+ \langle B_1(x,n,i)\rangle \\
+ \langle B_2(x,n,i)\rangle
+ \langle S_1(x,n,i)\rangle
+ \langle S_2(x,n,i)\rangle.
\end{multline}
Computing the expected values of individual  terms is based on the property that the
expected value of the product of independent random variables is equal to the
product of their expected values. For $[F_{140}^n(x)]_j$ we have
\begin{multline} \label{R140exp}
\langle [F_{140}^n(x)]_j \rangle =\langle \overline{x}_{j-1}x_{j} \rangle+ \left \langle \prod_{i=n-1}^{2n} x_{i-n+j} \right \rangle \\
\langle \overline{x}_{j-1}\rangle \langle x_{j} \rangle+  \prod_{i=n-1}^{2n} \langle  x_{i-n+j}  \rangle = (1-p)p + p^{n+1}.
\end{multline}
Similarly, for $D(x,n,j)$ we obtain
\begin{equation}
\langle D(x,n,j) \rangle =\langle x_{j-n} \rangle \prod _{m=j-n+1}^{j+1}\langle \overline{x}_{m} \rangle = p(1-p)^{n+1}.
\end{equation}
For $B_1(x,n,j)$ it is slightly more complicated,
\begin{multline*}
\langle B_1(x,n,j) \rangle=\mathcal{P} \left( n \right) \sum _{r=2}^{n} \left( \mathcal{P
} \left( r+1 \right) \langle x_{j-r} \rangle \langle x_{{j-r+1}} \rangle \langle x_{{j+2}}\rangle \prod _{m=j-r+2}^{j+1}\langle \overline{x}_{{m
}} \rangle \right) \\+\mathcal{P} 
\left( n+1 \right) 
\sum _{r=2}^{n}
 \left( \mathcal{P} \left( r \right) 
\langle \overline{x}_{{j-r}} \rangle 
\langle    x_{{j-r
+1}}\rangle  \langle x_{{j+2}}\rangle \prod _{m=j-r+2}^{j+1}\langle \overline{x}_{{m}}\rangle  \right)\\
= \mathcal{P} \left( n \right) \sum _{r=2}^{n}  \mathcal{P
} \left( r+1 \right) p^3 (1-p)^{r-1}  +\mathcal{P} 
\left( n+1 \right) 
\sum _{r=2}^{n}
  \mathcal{P} \left( r \right) 
(1-p)p^2 (1-p)^{r-1} 
\end{multline*}
Using $\mathcal{P}(r)=\frac{1}{2}(-1)^r+\frac{1}{2}$, the sums in the last line become
partial sums of geometric series which can be computed easily. The result, after simplification, becomes
\begin{multline}
\langle B_1(x,n,j) \rangle = \mathcal{P}\left( n+1 \right)  \left( {\frac { \left( 1-p \right) ^{3}{p
}^{n+1}p}{{p}^{2}-1}}-{\frac { \left( 1-p \right) ^{3}{p}^{3}}{{p}^{2}
-1}} \right) \\
+\mathcal{P}\left( n \right)  \left( {\frac { \left( 1-p
 \right) ^{2}{p}^{n+1}{p}^{2}}{{p}^{2}-1}}-{\frac { \left( 1-p
 \right) ^{2}{p}^{3}}{{p}^{2}-1}} \right).
 \end{multline}
 The terms involving $B_2$, $S_1$ and $S_2$ can be processed in the same fashion, yielding
 \begin{multline}
\langle B_2(x,n,j) \rangle = \mathcal{P} \left( n \right)  \left( {\frac {{p}^{2} \left( 1-p \right) ^
{n+1}}{p-2}}+{\frac {{p}^{2} \left( -1+p \right) ^{3}}{p-2}} \right) +\\
\mathcal{P} \left( n+1 \right)  \left( {\frac {p \left( 1-p \right) ^{n+2
}}{p-2}}+{\frac {p \left( -1+p \right) ^{3}}{p-2}} \right), 
\end{multline}
\begin{multline}
\langle S_1(x,n,j) \rangle =\mathcal{P}\left( n \right)  \left( {\frac {p \left( 1-p \right) ^{n+2}}
{p-2}}-{\frac {p \left( -1+p \right) ^{2}}{p-2}}+{\frac { \left( 1-p
 \right) ^{n+2}}{p-2}}-{\frac { \left( -1+p \right) ^{4}}{p-2}}
 \right) \\
 +\mathcal{P}\left( n+1 \right)  \left( {\frac {p \left( 1-p
 \right) ^{n+1}}{p-2}}-{\frac {p \left( -1+p \right) ^{2}}{p-2}}+{
\frac { \left( 1-p \right) ^{n+3}}{p-2}}-{\frac { \left( -1+p \right) 
^{4}}{p-2}} \right),
\end{multline}
and
\begin{multline}\label{S2expr}
\langle S_2(x,n,j)\rangle = \mathcal{P}\left( n \right)  \left( -{\frac {{p}^{n+2} \left( 2\,p-1
 \right) }{{p}^{2}-1}}+{\frac {{p}^{n+4}}{{p}^{2}-1}}-{\frac { \left( 
-p \right) ^{n+4}}{p+1}}-{\frac {{p}^{4} \left( p-2 \right) }{p+1}}
 \right) \\
 + \mathcal{P}\left( n+1 \right)  \left( {\frac {{p}^{n+3}}{{p}^{
2}-1}}-{\frac {{p}^{n+3} \left( 2\,p-1 \right) }{{p}^{2}-1}}-{\frac {
 \left( -p \right) ^{n+4}}{p+1}}-{\frac {{p}^{4} \left( p-2 \right) }{
p+1}} \right) 
. 
\end{multline}\label{probsol1}
After combining all expressions given by eqs. (\ref{R140exp}--\ref{S2expr}) and simplifying the result, the final formula for 
the expected value of a cell after $n$ iterations becomes
\begin{multline}
\langle [F^n_{156}(x)]_i \rangle= \frac{1}{2}+{\frac {{p}^{n+3}}{1+p}}+{\frac { \left( 1-p \right) ^{n+3}}{p-2}}
+\frac{1}{2}\,{\frac { p\left( p-1 \right)  \left( 2\,p-1 \right)  \left( -1
 \right) ^{n}}{ \left( 1+p \right)  \left( p-2 \right) }}.
\end{multline}
Although it is not obvious at the first sight, this formula exhibits certain
level of symmetry. If we introduce $q=1-p$, representing probability of 0's in the
initial configuration, then it takes the form
$$
\langle [F^n_{156}(x)]_i \rangle=
\frac{1}{2} + \frac{p^{n+3}}{1+p} -\frac{q^{n+3}}{1+q}
+ \frac{pq(p-q)}{2(1+p)(1+q)} (-1)^n.
$$
One can see that interchange of $p$ and $q$ changes the value of this expression to
$1-\langle [F^n_{156}(x)]_i \rangle$.
This echoes the fact that Boolean conjugation (interchange of roles of 0's and 1's in the definition of the local function) combined with spatial reflection does not change the rule
156.

The presence of the term with $(-1)^n$ in the expression for $\langle [F^n_{156}(x)]_i \rangle$ causes oscillations of the density of ones, and a quick look at
Figure \ref{rule156patt} makes it clear that they appear due to ``blinkers'' developing in the spatiotemporal pattern. These blinkers develop more of less easily depending on the
value of the initial density $p$. The amplitude of oscillations, which we can define as the absolute value of the coefficient in front of $(-1)^n$,  is given by
$$
A(p)=\frac{1}{2}\, \left| {\frac { p\left( p-1 \right)  \left( 2\,p-1 \right)}{ \left( 1+p \right)  \left( p-2 \right) }} \right|.
$$
For what value of $p$ are these oscillations strongest? Graph of $A(p)$ vs. $p$ 
shown in Figure~\ref{Apgraph} reveals that this happens at two values which can be obtained by solving $dA/dp=0$ for $p$,  that is,
$$
\frac{dA}{dp}=
{\frac {{p}^{4}-2\,{p}^{3}-5\,{p}^{2}+6\,p-1}{ \left( 1+p \right) ^{2}
 \left( p-2 \right) ^{2}}} =0.
$$
\begin{figure}[b]
\begin{center}
\includegraphics[scale=0.9]{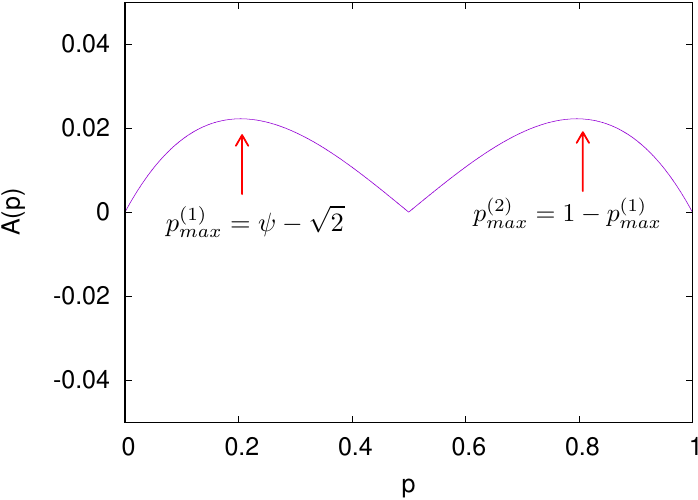}
\end{center}
\caption{Graph of the amplitude of density's oscillations  as a function of the initial density $p$.}\label{Apgraph}
\end{figure}
Two solutions of the above are in the interval $[0,1]$, namely
\begin{align*}
p_{max}^{(1)}&=\psi - \sqrt{2}\approx 0.2038204260, \\
p_{max}^{(2)}&=1-p_{max}^{(1)} \approx .796179574,
\end{align*}
where 
$$\psi=\frac{1+\sqrt{5}}{2}\approx 1.618033988$$
is the golden ratio.
\begin{figure}
\begin{center}
\includegraphics[scale=0.6]{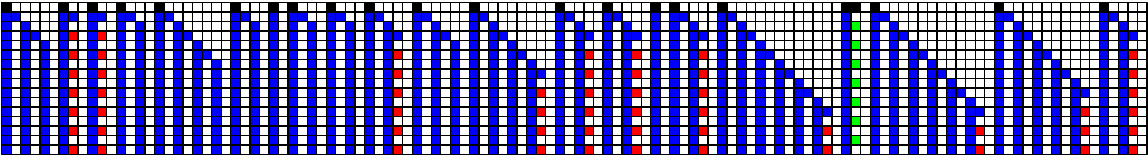}\\
$p=p_{max}^{(1)}$ \\[1em]
\includegraphics[scale=0.6]{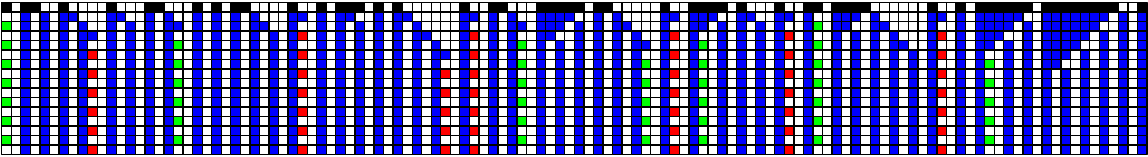}\\
$p=1/2$ \\[1em]
\includegraphics[scale=0.6]{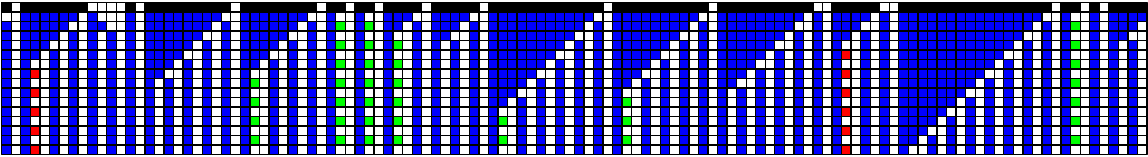}\\
$p=p_{max}^{(2)}$
\end{center}
\caption{Spatiotemporal patterns of rule 156 generated starting from random initial configurations with three different values of the initial density $p$. Sites in state 1
are black in the initial string, red in odd blinkers and green in even blinkers. Other sites in state 1 are blue. Lattice of 120 sites with periodic boundaries.
}\label{blinkersfig}
\end{figure}
When $p=1/2$, the amplitude is minimal and equal to zero (similarly as in the trivial cases of
$p=0$ and $p=1$).
Figure \ref{blinkersfig} shows examples of spatiotemporal patterns generated
by rule 156 for $p=p_{max}^{(1)}$, $p=1/2$ and  $p=p_{max}^{(2)}$.
Blinkers which are in state 1 at odd times are called odd and shown in red color,
while those which are in state 1 at even times are called even and are shown in green
color. One can see that at $p_{max}^{(1)}$ odd blinkers prevail, even though even ones
occasionally appear as well. Similarly, at $p_{max}^{(2)}$ even blinkers dominate.
When $p=1/2$, both types of blinkers are present, but they occur with equal frequency,
thus they cancel each other's oscillations and on average the amplitude of oscillations is zero.

The amplitude of oscillations at  $p_{max}^{(1,2)}$ is rather small,
$$A\left(p_{max}^{(1,2)} \right)=
\frac{4}{3}\sqrt{2} -  \frac{5}{6} \sqrt{5} \approx 0.022228101,
$$
meaning that it would be rather difficult to discover this phenomenon 
by computer experiments alone, not to mention that the value
of $p_{max}^{(1)}=\psi-\sqrt{2}$ is far from obvious.
 This highlights the fact that obtaining 
exact solutions of the initial value problem for various CA rules can 
potentially lead to discovering their novel properties,
thus it is a worthwhile direction of research.

Using the method outlined in this section one can
obtain probability of occurrence of any block $\mathbf{a}$ after $n$ iterations
of rule 156 starting from random initial condition described earlier,
where all values of $x_i$ are initially set as independent and identically distributed random variables, such that $Pr(x_i=1)=p$ and $Pr(x_i=0)=1-p$, where $p\in[0,1]$ is a fixed
parameter. Let us denote such probability by $P_n(\mathbf{a})$.
We already computed $P_n(1)=\langle [F^n_{156}(x)]_i \rangle$, hence
$P_n(0)=1-\langle [F^n_{156}(x)]_i \rangle$, yielding
\begin{multline}
P_n(0)= \frac{1}{2}-{\frac {{p}^{n+3}}{1+p}}-{\frac { \left( 1-p \right) ^{n+3}}{p-2}}
-\frac{1}{2}\,{\frac { p\left( p-1 \right)  \left( 2\,p-1 \right)  \left( -1
 \right) ^{n}}{ \left( 1+p \right)  \left( p-2 \right) }}.
\end{multline}
If one wants to compute, for example, $P_n(00)$, one just needs to note that block
$00$ will appear at position $i$ and $i+1$ when $[F^n_{156}(x)]_i=0$
and $[F^n_{156}(x)]_{i+1}=0$, that is, when 
$$\left( 1- [F^n_{156}(x)]_i \right) \left(1-[F^n_{156}(x)]_{i+1}) \right) = 1.$$
Due to translation invariance we then have
$$
P_n(00)=\Big \langle \left( 1- [F^n_{156}(x)]_0 \right) \left(1-[F^n_{156}(x)]_{1}) \right)   \Big \rangle.
$$
This expected value can be now computed using the expression for
$[F^n_{156}(x)]_i$ given in 
eq. (\ref{rule156solution}) in a similar way as we did for $\langle [F^n_{156}(x)]_i \rangle$. The calculation, however, are rather long and  tedious and are best done using a
computer algebra system, thus we will not reproduce them here. We just 
give the final result,
\begin{align*}
P_n(00)&=\frac{1}{2}\,{\frac {p \left( p-1 \right)  \left( 2\,{p}^{4}-4\,{p}^{3}+2\,{p}
^{2}+1 \right) }{ \left( 1+p \right)  \left( p-2 \right) }}-\frac{1}{2}\,{
\frac { \left( p-1 \right) {p}^{n+2}}{1+p}}
\\
&+\frac{1}{2}\,{\frac { \left( 
3\,p-4 \right)  \left( 1-p \right) ^{n+2}}{p-2}}
+\frac{1}{2}\,{\frac { \left( 
2\,p-1 \right)  \left( p-1 \right)  \left( -p \right) ^{n+2}}{1+p}}
\\
&-\frac{1}{2}\,{\frac {p \left( 2\,p-1 \right)  \left( p-1 \right) ^{n+2}}{p-2}}
-\frac{1}{2}\,{\frac { p \left( p-1 \right)  \left( 2\,p-1 \right)  \left( -1
 \right) ^{n}}{ \left( 1+p \right)  \left( p-2 \right) }}.
\end{align*}
We can see that oscillations similar to what we saw in $P_n(1)$ are present
in $P_n(00)$, these are identified by $(-1)^n$ factor in the last term. In addition, there also damped oscillations corresponding to the term with $(-p)^{n+1}$ (fourth term)
and $(p-1)^{n+2}$ (fifth term). These quickly tend to zero for $p \in (0,1)$, 
thus they would also not be easily identified in numerical simulations of rule 156.

Of course probabilities of other blocks can be also obtained in a similar fashion, at least in principle, as long as we are willing to compute the relevant expected values, which
can became dauntingly complex for longer blocks. We show below solutions
for $P_n(000)$ and $P_n(010)$, obtained with the help of Maple symbolic algebra software.
\begin{align*}
P_n(000)&= \left( 1-p \right) ^{n+3}
,\\\nonumber
P_n(010)&=-\frac{1}{2}\,{\frac {4\,{p}^{6}-12\,{p}^{5}+12\,{p}^{4}-4\,{p}^{3}+{p}^{2}-p+
2}{ \left( 1+p \right)  \left( p-2 \right) }}+{\frac { \left( {p}^{2}-
p-1 \right) {p}^{n+2}}{1+p}}
\\
&+{\frac { \left( 1-p \right) ^{n+2}}{p-2}}
-{\frac { \left( 2\,p-1 \right)  \left( p-1 \right)  \left( -p
 \right) ^{n+2}}{1+p}}+{\frac {p \left( 2\,p-1 \right)  \left( p-1
 \right) ^{n+2}}{p-2}}
 \\
 &-\frac{1}{2}\,{\frac {p \left( p-1 \right)  \left( 2\,p-
1 \right)  \left( -1 \right) ^{n}}{ \left( 1+p \right)  \left( p-2
 \right) }}
. \end{align*}
Having $P_n(0)$, $P_n(00)$, $P_n(000)$ and $P_n(010)$, one can compute probabilities of all other blocks of length up to 3, using consistency conditions and formulae derived
in \cite{paper50,hfbook},
\begin{align} \label{shortform3}
\left[ \begin {array}{c} 
P(001)\\
P(011)\\  
P(100) \\
P(101)\\
P(110)\\
P(111)
\end {array} \right] &= 
\left[ \begin {array}{c} 
P(00)-P(000) \\ 
P(0)-P(00) -P(010) \\  
P(00)-P(000) \\  
P(0)-2 P(00)+P(000) \\ 
P(0)-P(00)-P(010)\\  
 1-3 P(0) +2 P(00) +P(010) 
\end {array} \right]. \nonumber \\
\left[ \begin {array}{c} 
P(01)\\
P(10)\\  
P(11)
\end {array} \right] &= 
\left[ \begin {array}{c} 
P(0) -P(00) \\
P(0) -P(00) \\ 
 1-2 P(0)+P(00)
\end {array} \right], \nonumber  	\\
 P(1) &= 1-P(0). 
\end{align}
For blocks $\mathbf{a}$ of higher length similar calculations could be performed, thus
the probability measure resulting from $n$ iterations of the Bernoulli measure could  theoretically be described with arbitrary precision (by giving probabilities of blocks of any length). In practice, however, as already remarked, going beyond blocks of three symbols becomes 
very cumbersome due to the complexity of the resulting expressions, even if one uses a computer algebra system.

\section{Conclusions}
The method for solving  the IVP outlined here is applicable to other rules as well. 
In the case of rule 156, the ``perturbation term'' was non-negative, but
in other cases it may be necessary not only to add some terms, but also subtract.
Such situation occurs in rule 78, for which
$$f_{78}(x_1,x_2,x_3)=f_{206}(x_1,x_2,x_3)- x_1x_2x_3.$$
Solution of the IVP for rule 206 is known,
thus  the IVP for rule 78 can be solved as well \cite{hfbook}.

A natural question to ask is how general can the method be. Although 
at the moment the definitive answer is not known, it is worth noting that rule 156
(as well as the aforementioned rule 78) is almost equicontinuous\footnote{
The word 01 is 1-blocking for rule 156, and 10 is 1-blocking for rule 78,
thus using the result of K{\r{u}}rka  \cite{Kurka1997}, these rules are almost equicontinuous.}
,
and many almost-equicontinuous rules  seem to produce patterns amenable to decomposition into relatively simple elements.
In fact, the IVP is solvable for almost all elementary rules possessing some equicontinuity property (equicontinuous, almost equicontinuous, or with equicontinuous or almost equicontinuous direction), yet in most cases other (simpler) methods can be used to obtain the solution~\cite{hfbook}.

Another question is how to split the local function $f$ into an unperturbed rule $g$ and 
perturbation $b$, as we did in eq. (\ref{local156})? The crucial property 
is not only that $g$ must be solvable but also that the perturbation must be relatively ``mild'', so that adding $b$ does not destroy
the pattern of $g$ but modifies it in such a way that the changes can be described
by a simple algorithm. Perturbation $b$ changing only one bit in the rule table
does not always guarantee this, as sometimes changing one bit completely modifies the nature of the rule. For example,
$$f_{18}(x_1,x_2,x_3)=f_{19}(x_1,x_2,x_3) - \overline{x}_1 \overline{x}_2 \overline{x}_3,
$$
thus rules 18 and 19 differ only on $(x_1,x_2,x_3)=(0,0,0)$. 
\begin{figure}
\begin{center}
\includegraphics[scale=0.3]{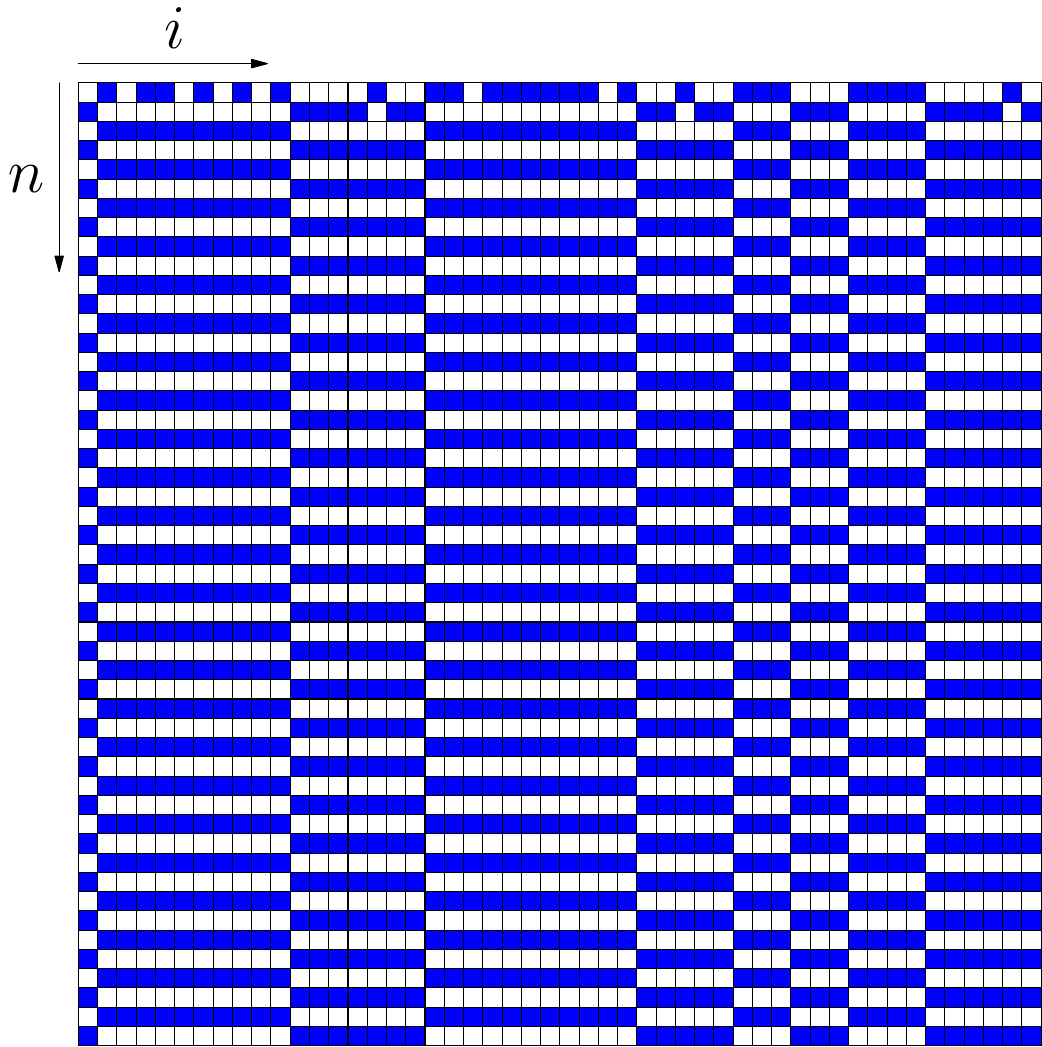}
\includegraphics[scale=0.3]{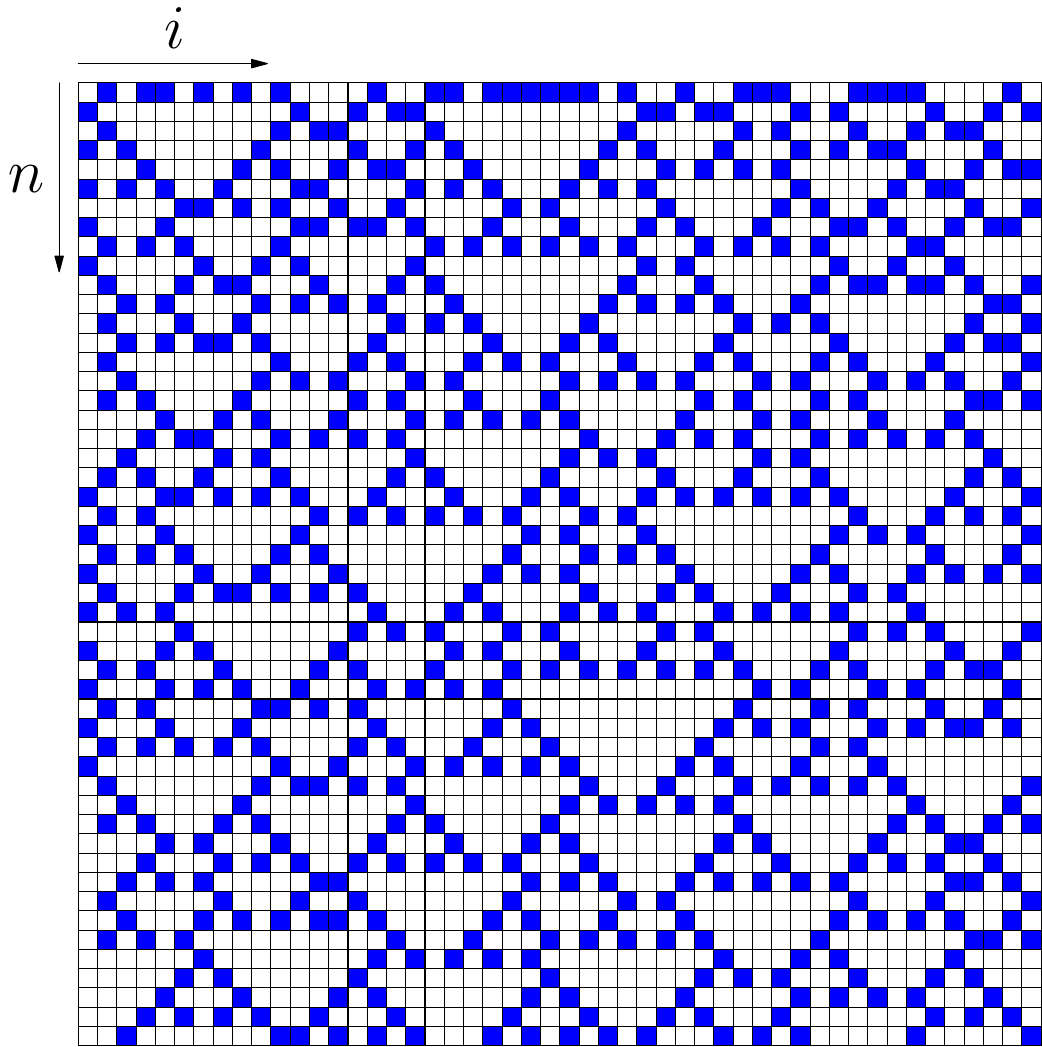}
\end{center}
\caption{Spatiotemporal patterns of rule 19 (left) and rule 18 (right).}
\label{rule19-18}
\end{figure}
In spite of this,  although
rule 19 is easily solvable \cite{hfbook}, rule 18 
is most likely not solvable, as it exhibits complex fractal-like backgound with 
``defects''
which  perform a pseudo-random walk  and annihilate \index{annihilation of defects} upon collision  \cite{GRASSBERGER198452,Eloranta1992}. Comparing patterns generated by these
two rules makes it rather clear at the first sight (cf. Figure \ref{rule19-18}).

\vskip 1em
\noindent\textbf{Acknowledgment} The author acknowledges partial financial support from the Natural Sciences and Engineering Research Council of Canada in the form of Discovery Grant
RGPIN-2015-04623.

\end{document}